\documentclass[pra,twocolumn,floatfix,notitlepage,groupedaddress,longbibliography]{revtex4-1}

\usepackage[hidelinks,colorlinks=false,citecolor=black,linkcolor=black]{hyperref}
\usepackage{mathtools}
\usepackage{subcaption}
\usepackage{amssymb,amsmath,amstext,amsthm}
\usepackage{graphicx}
\usepackage{epstopdf}
\usepackage{color}
\usepackage{bm}
\usepackage{appendix}
\usepackage[T1]{fontenc}
\usepackage{bbm}
\usepackage{latexsym}

\usepackage{algorithm,algorithmic}

\usepackage{float}
\usepackage{verbatim}
\usepackage{lipsum}
\usepackage{braket}
\usepackage{ulem}

\newcommand{\del}[0]{\partial}

\newtheorem{theorem}{Theorem}

\begin{document}

\title{Globally optimal  interferometry with lossy twin Fock probes}
\author{T.J.\,Volkoff}
\affiliation{Theoretical Division, Los Alamos National Laboratory, Los Alamos, NM, USA.}
\author{Changhyun Ryu}
\affiliation{MPA-Q, Los Alamos National Laboratory, Los Alamos, NM, USA.}

\begin{abstract}
Parity or quadratic spin (e.g., $J_{z}^{2}$) readouts of a Mach-Zehnder (MZ) interferometer probed with a twin Fock input state allow to saturate the optimal  sensitivity attainable among all mode-separable states with a fixed total number of particles, but only when the interferometer phase $\theta$ is near zero. When more general Dicke state probes are used, the parity readout saturates the quantum Fisher information (QFI) at $\theta=0$, whereas better-than-standard quantum limit performance of the $J_{z}^{2}$ readout is restricted to an $o(\sqrt{N})$ occupation imbalance. We show that a method of moments readout of two quadratic spin observables $J_{z}^{2}$ and $J_{+}^{2}+J_{-}^{2}$ is globally optimal for Dicke state probes, i.e., the error saturates the QFI for all $\theta$. In the lossy setting, we derive the time-inhomogeneous Markov process describing the effect of particle loss on twin Fock states, showing that method of moments readout of four at-most-quadratic spin observables is sufficient for globally optimal estimation of $\theta$ when two or more particles are lost. The analysis culminates in a numerical calculation of the QFI matrix for distributed MZ interferometry on the four mode state $\ket{{N\over 4},{N\over 4},{N\over 4},{N\over 4}}$ and its lossy counterparts, showing that an advantage for estimation of any linear function of the local MZ phases $\theta_{1}$, $\theta_{2}$ (compared to independent probing of the MZ phases by two copies of $\ket{{N\over 4},{N\over 4}}$) appears when more than one particle is lost.
\end{abstract}
\maketitle

\section{Introduction}
\label{sec:intro}
In the context of optical interferometry with nonclassical states of light, the optical twin Fock (TF) state was introduced as a candidate probe state that minimizes phase fluctuations between the arms of an interferometer \cite{PhysRevLett.71.1355}. The operating principle of the interferometer is that although the input TF state has maximal phase uncertainty prior to entering a Mach-Zehnder (MZ) interferometer, its phase variance after the first beamsplitter is lower by a factor of $O(I^{-1})$ compared to a probe state consisting of single-mode lasers of the same total intensity $I$ per measurement interval (lowering the optical standard quantum limit sensitivity of $O(I^{-1})$ to optical Heisenberg limit sensitivity $O(I^{-2})$). For present-day laser-based gravitational wave detectors, suppression of phase and amplitude fluctuations across the operating spectrum motivates the use of novel non-classical electromagnetic field probe states.  In the LIGO interferometer, for example, phase fluctuations in the quantum electromagnetic field give rise to the photonic shot noise that constitutes the primary limitation to sensitivity in the high-frequency domain. Fluctuations in the amplitude quadrature give rise to the radiation pressure noise that constitutes the primary limitation to sensitivity in the low-frequency domain \cite{PhysRevD.23.1693,PhysRevX.13.041021}. Therefore, preparation of a probe field with frequency-dependent quadrature squeezing enables globally improved noise spectral densities \cite{PhysRevLett.124.171102,PhysRevX.13.041021}. However, even at fixed wavelength, state of the art optical TF states with $O(1)$ photons are produced only probabilistically \cite{opticalfock}, and not with intensities large enough to be relevant for application in LIGO. 

On the other hand, matter-wave interferometry with ensembles of phase-coherent atoms provides a alternative framework for gravitational wave detection  \cite{PhysRevD.78.122002,Abe_2021}. Analogously to the optical TF state, a TF state of massive bosons produced by strong repulsion of atoms in a double-well optical trap could potentially improve the sensitivity of such matter-wave interferometers beyond the atomic standard quantum limit ($1/N$ scaling of the optimal estimator variance, where $N$ is the number of atoms). Whereas traditional atom interferometers apply optimal Bragg splitting schemes to non-entangled states such as a Bose-Einstein condensate to produce an optimal non-entangled probe state \cite{PhysRevLett.67.181,PhysRevA.71.043602}, the input TF state exhibits particle entanglement \cite{PhysRevLett.91.097902} and remains entangled throughout the MZ interferometer sequence. As far as experimental atom interferometry is concerned, the sensitivity boost provided by TF states can allow smaller MZ loops to be utilized in the interferometry sequence. TF states of neutral atoms in an optical dipole trap have been prepared with $O(10^{4})$ atoms \cite{tmw}. 

The present work focuses on the parameter estimation setting defined by a Mach-Zehnder (MZ) interferometer \cite{PhysRevA.33.4033}, which applies the operation $e^{i{\pi\over 2}J_{x}}e^{-i\theta J_{z}}e^{-i{\pi\over 2}J_{x}}=e^{-i\theta J_{y}}$  to the probe state $\rho$, which is generally a mixed state. From the parametrized state $\rho_{\theta}:= e^{-i\theta J_{y}}\rho e^{i\theta J_{y}}$, the estimate $\tilde{\theta}$ of the parameter $\theta$ is calculated, and it satisfies the one-shot quantum Cram\'{e}r-Rao inequality
\begin{equation}
(\Delta \tilde{\theta})^{2}\big\vert_{\theta} \ge {1\over \mathcal{F}_{\theta}}
\label{eqn:qcrb}
\end{equation} 
where $\mathcal{F}_{\theta}$ is the quantum Fisher information (QFI) based on the symmetric logarithmic derivative \cite{holevobook}. The relevant expressions for the QFI depend on whether the probe state $\rho$ is pure or mixed, and we refer to \cite{yu} and references therein for explicit formulas. A globally optimal quantum sensing protocol produces an estimator $\tilde{\theta}$ that saturates (\ref{eqn:qcrb}) at all parameter values $\theta$.

In the remainder of the present section, we review relevant background on quantum sensing aspects of twin Fock states. In Section \ref{sec:tt}, we move on to considering the possibility of saturating the inequality (\ref{eqn:qcrb}) when applying practical readout schemes to pure, but imperfectly prepared, twin Fock states. We generalize the method of moments error of a single readout observable to a generalized signal-to-noise ratio that quantifies correlated errors of a list of noncommuting observables. Observables are identified for which the generalized signal-to-noise ratio globally saturates the QFI (i.e., saturates the QFI for all $\theta$). Sections \ref{sec:lossytf} and \ref{sec:grad} respectively identify lists of observables associated with a generalized signal-to-noise ratio which globally saturates the QFI for the case of lossy twin Fock probe states and for instances of lossy, spatially extended twin Fock probe states. The spatially extended twin Fock states probe multiple spatially separated interferometric phase shifts, which define a multiparameter sensing setting relevant to, e.g., spatially resolved gravimetry and magnetometry. Section \ref{sec:disc} summarizes the results and discusses directions for future research.

We now define the TF state and provide background on their known properties related to quantum metrology. The TF state of $N$ bosonic atoms distributed between two orthogonal single particle modes $\ket{0}$, $\ket{1}$ is the zero $J_{z}$ weight vector in a spin-$N/2$ representation of SU(2). It is the $n={N\over 2}$ case of the more general $N$-particle Dicke state often written as
\begin{equation}
    \ket{N-n,n}:={1\over \sqrt{{N\choose m}}}\sum_{\text{Ham}(x)=n}\ket{x_{1}}_{1}\otimes \cdots\otimes \ket{x_{N}}_{N}
    \label{eqn:tf}
\end{equation}
where $x=x_{1}\ldots x_{N}\in \lbrace 0,1\rbrace^{N}$ is a binary string identified with a computational basis element of $N$ qubits,  and $\text{Ham}(x)$ is its Hamming weight. The notation in the left hand side owes to the Schwinger boson form of the $J_{z}$ operator, $2J_{z}=a^{\dagger}a-b^{\dagger}b$ so that $J_{z}\ket{N-n,n}=({N\over 2}-n)\ket{N-n,n}$. The state (\ref{eqn:tf}) is \textit{mode-separable} because the first-quantized state is actually entangled: if the particles themselves are bipartitioned in any way (i.e., if the $N$ qubits are bipartitioned in any way in (\ref{eqn:tf})), and one of the partitions is traced over, the resulting state is impure.  

Our model of the MZ interferometer is the standard one: a two mode, mode-separable bosonic state $\ket{\psi}$ is parametrized by a phase according to
\begin{equation}
\ket{\psi}\mapsto \ket{\psi(\theta)}=e^{i{\pi\over 2}J_{x}}e^{-i\theta J_{z}}e^{-i{\pi\over 2}J_{x}} = e^{-i\theta J_{y}}\ket{\psi}.
\label{eqn:tttt}
\end{equation}
Eq.(\ref{eqn:tttt}) is an example of the shift model of parameter estimation \cite{holevobook}.
The TF state possesses the following optimality property for sensing the phase difference in the arms of a MZ interferometer: it has the largest quantum Fisher information on the unitary path generated by $J_{y}$ over all mode-separable probe states with a fixed total number of particles $N=a^{\dagger}a+b^{\dagger}b$, $N$ even \cite{PhysRevA.90.025802}. The operator $J_{y}$ generates the phase difference dynamics of a MZ interferometer.

As pointed out by Lang and Caves, mode-separable states (or product states in the setting of optical MZ interferometry) are the only sensible input states to consider in single-parameter MZ interferometry. The whole purpose of the first beamsplitter is to generate the quantum coherence required for the phase difference sensing task. If one has access to mode-entangled states at the input of the MZ interferometer, an optimal input state is a Greenberger-Horne-Zeilinger (GHZ) state (recall that for $N$ atoms, the GHZ states are a class of states defined as an equal amplitude superposition of the highest and lowest eigenvector of a total spin operator $\vec{n}\cdot \vec{J}$ in the spin-$N/2$ representation of $\mathfrak{su}(2)$, where $\vec{n}\in\mathbb{R}^{3}$).  On the other hand, in the setting of multiparameter estimation of an SU(2) element, the TF state has a similar basic optimality property as the GHZ state has in single-parameter MZ interferometry. Consider the parametrized unitary dynamics $U(\theta)=e^{-i\theta_{1}J_{x}-i\theta_{2}J_{y}}$ and task of minimizing the sum of the mean squared errors of estimates of $\theta_{1}$ and $\theta_{2}$. This total error is bounded below by $\text{tr}(\mathcal{F}^{-1})$, where $\mathcal{F}$ is the QFI matrix \cite{PhysRevA.95.012305}. This quantity, in turn, is lower bounded by ${1\over \text{Var}J_{x} + \text{Var}J_{y}}$, with equality achieved on pure states with a reflection symmetry about $J_{x}$ or $J_{y}$. The minimal value is obtained when: 1. the variances are equal, 2. the state is centered in the spherical phase space (i.e., $\langle J_{x}\rangle = \langle J_{y}\rangle =\langle J_{z}\rangle=0$), and 3. $\langle J_{z}^{2}\rangle=0$, all three of which are satisfied by the TF state, with the last condition being unique to the TF state \cite{genoni}. Fujiwara's condition for achievability of the QFI bound is also satisfied \cite{fuji}, so there is an estimator that obtains this noise value. Indeed, method of moments estimation of $\vec{\theta}$ by readout of the nonlinear observables $O_{j}\ket{{N\over 2},{N\over 2}}\bra{{N\over 2},{N\over 2}}O_{j}$, with $O_{1}=J_{x}$, $O_{2}=J_{y}$, allow to saturate quantum Fisher information (QFI) matrix for $\vec{\theta}\rightarrow (0,0)$ \cite{Gessner2020}. Saturation of the QFI matrix can also be obtained using four operators which are at-most-quadratic in the spin operators \cite{Fadel_2023}.

Recent work has shown that Dicke states $\ket{N-m,m}$ can be prepared with $O(m\log {N\over m})$-depth quantum circuits that allow two-qubit operations between any qubit registers \cite{bart}. Analog schemes for generating approximate TF states of a system of double-well bosons include implementation of the two-axis countertwisting Hamiltonian \cite{PhysRevA.92.013623}, and quantum alternating operator ansatz circuits that approximately convert a tensor product state to the TF state for large $N$ \cite{PhysRevA.95.062317}. Experimentally, small twin Fock states have also been generated in internal states of spinor Bose-Einstein condensates \cite{apell}, and in photons in orthogonal polarization modes \cite{PhysRevLett.107.080504}. Optimal method of moments readouts for specific single-parameter estimation problems, including MZ interferometry, with probe states consisting of coherent states and spin squeezed states of spin-1 Bose-Einstein condensates have been analyzed in \cite{PhysRevResearch.5.013087,Niezgoda_2019,PhysRevA.104.042415,Mao2023}.

\section{Pure imperfect TF probes\label{sec:tt}}

Unlike other entangled probe states of two-mode bosons such as one-axis twisted states \cite{PhysRevA.47.5138}, the physical origin of the sensitivity of TF states for interferometry is not spin-squeezing \cite{PhysRevLett.112.155304}. Rather, it is the property of low uncertainty in the difference of phase (near a phase difference of zero or $\pi$) appearing between the paths of the MZ interferometer after the TF state undergoes the first beamsplitter \cite{PhysRevA.68.023810}. In fact, the Dicke state amplitudes of the $\pi/2$-rotated TF state are greatest on the $\ket{N,0}$ and $\ket{0,N}$ states, similar to an approximate  GHZ state. Although a GHZ state in the arms of the MZ interferometer allows to obtain the minimal uncertainty of an estimator of the MZ phase, the input state required to generate such a configuration is simply a rotated GHZ state, which is clearly not any easier to generate. By contrast, in the case of ultracold atoms, the TF state can be prepared as the ground state of a $N$ atoms in a double-well optical trap \cite{PhysRevA.90.063601,Grond_2010}, which makes it practically advantageous compared to a GHZ probe. Further, loss of a single particle from a GHZ state renders it useless for interferometry beyond $O(N)$ scaling of the QFI (i.e., standard quantum limit (SQL) scaling). Photon-loss-robust optical states for MZ interferometry were explored as GHZ alternatives in Ref.\cite{PhysRevLett.102.040403}.

In this section, we consider the case of pure, but imperfectly prepared, twin Fock states. Such states coincide with sequences of Dicke states of the form $\ket{{N\over 2}\pm m, {N \over 2}\mp m}$ for $m$ scaling as $o(N)$.  The quantum Fisher information (QFI)  $\mathcal{F}$ for MZ interferometry using an arbitrary Dicke state (for even $N$) $\ket{{N\over 2}+m,{N\over 2}-m}$ is given by  \cite{RevModPhys.90.035005}
\begin{equation}
\mathcal{F}(m)= 4\text{Var}J_{y}={N^{2}\over 2}-2m^{2}+N,
\label{eqn:qqfi}
\end{equation} independent of $\theta$. The particle number dependence in (\ref{eqn:qqfi}) indicates the possibility of Heisenberg scaling, i.e., $O(N^{2})$, but not the Heisenberg limit $N^{2}$, which is only achievable by a GHZ probe state. In fact, this value of the quantum Fisher information is attained for any spin generator $\vec{n}\cdot \vec{J}$ with $\vec{n}$ a unit vector in the $xy$-plane \cite{PhysRevA.82.012337}. Note that when $m=\pm {N\over 2}$, this expression for the QFI also covers the standard quantum limit (SQL) value $N$. One concludes that Dicke states exhibit asymptotic Heisenberg scaling for interferometry if  the magnitude of the $J_{z}$ component scales as $\lambda N$ with $\lambda <1/2$.

To locally saturate the QFI, we first review two practical measurement schemes that have appeared in the literature. A Bayesian estimation scheme locally achieving the same scaling as the QFI was discussed in Refs. \cite{PhysRevLett.71.1355, Meiser_2009}. However, a single mode parity measurement (we will use $(-1)^{b^{\dagger}b}$ as the single-mode parity operator) exhibits the same scaling while having practical advantages \cite{PhysRevA.54.R4649,PhysRevA.61.043811,PhysRevLett.92.209301,PhysRevA.68.023810}. A simple calculation using the small angle asymptotics of Wigner's little-$d$ function allows to show that with the parity signal \begin{equation}g(\theta):= \langle (-1)^{b^{\dagger}b}\rangle_{e^{-i\theta J_{y}}\ket{{N\over 2}+m,{N\over 2}-m}},\label{eqn:fifi}\end{equation} the method of moments error $(\Delta \tilde{\theta})^{2}$ for extracting an estimate $\tilde{\theta}$ of $\theta$ is
\begin{equation}
    (\Delta \tilde{\theta})^{2}\vert_{\theta=0}={1- g(\theta)^{2}\over g'(\theta)^{2}}\big\vert_{\theta=0}= \left( {N^{2}\over 2}-2m^{2}+N\right)^{-1}
    \label{eqn:parpar}
\end{equation}
where we recall the formula for the method of moments error \begin{equation}(\Delta\tilde{\theta})^{2}={\text{Var}_{\ket{\psi(\theta)}}A\over \left( {d\langle A \rangle_{\ket{\psi(\theta)}}\over d\theta}\right)^{2}}\label{eqn:methmom}\end{equation} of a phase estimator $\tilde{\theta}$ obtained from measurement of observable $A$. The reciprocal QFI $1/\mathcal{F}$ is a lower bound to (\ref{eqn:methmom}) for all $\theta$.
Eq.(\ref{eqn:parpar}) implies that at $\theta=0$, the method of moments error of $\tilde{\theta}$ obtained from the parity signal Eq.(\ref{eqn:fifi}) saturates the reciprocal QFI for the MZ interferometer, regardless of which Dicke state is taken as input. The error of the method of moments  estimator $\tilde{\theta}$ at arbitrary angles is shown in Fig. \ref{fig:f1} for $N=64$, $m=0, 4, 8, 16$, in which is it seen that the ideal TF state has a revival of saturating the QFI near $\theta=\pi/2$. The method of moments error for a parity observable also saturates the QFI for the GHZ state \cite{PhysRevLett.96.010401}, but globally over all possible phases \cite{PhysRevA.68.023810}. From the plots, we conclude that for small $\theta$, the effect of the interferometer is to change the expected parity of the computational basis support of the rotated Dicke state.

It should be noted that a parity measurement requires single-particle resolution at a detector. For high intensity optical systems or high-occupation two-mode bosonic system, such an idealistic measurement is not always feasible. By constrast, it would be desirable to use low moments of a measurement of a total spin operator to extract a high-precision estimator.  It has been shown that for a TF state input, the method of moments error for an estimator obtained from measurement of $J_{z}^{2}$ saturates the QFI at $\theta=0$ \cite{PhysRevA.57.4004}. 
However, unlike the parity measurement, the $J_{z}^{2}$ measurement does not saturate the Dicke state QFI for arbitrary $m$, and it remains to consider the question of how robust the Heisenberg scaling of this estimation strategy is when a Dicke state is prepared. The following theorem 

\begin{theorem}
\label{thm:ooo}
Let $\mathcal{F}(m)$ be the QFI for a MZ interferometer with input Dicke state $\ket{{N\over 2}+m,{N\over 2}-m}$, and let $\Delta \tilde{\theta}$ be the method of moments error Eq. (\ref{eqn:methmom}) for a $J_{z}^{2}$ measurement. Then,
\begin{equation}
\lim_{N\rightarrow \infty} {(\Delta \tilde{\theta})^{-2} \vert_{\theta=0}\over \mathcal{F}(m)} ={1\over 4m^{2}+1}.
\label{eqn:asas}
\end{equation}
If $m$ scales with $N$ as $m=o(\sqrt{N})$, then $\lim_{N\rightarrow \infty}N(\Delta \tilde{\theta})^{2}\vert_{\theta=0} = 0$.
\end{theorem}

\begin{proof}
We use the moments of $J_{z}$ in Appendix \ref{sec:appa} to find that at $\theta=0$, the method of moments error $(\Delta \tilde{\theta})^{2} \vert_{\theta=0}$ is given by
\begin{widetext}
\begin{equation}
    (\Delta \tilde{\theta})^{-2} \vert_{\theta=0}={ \sum_{u=0}^{1}(1+(-1)^{u}2m)^{2}({N\over 2}+(-1)^{u}m+1)({N\over 2}-(-1)^{u}m)\over 4\left( {N^{2}\over 4} - 3m^{2}+{N\over 2}\right)^{2}}.
    \label{eqn:rrr}
\end{equation}
\end{widetext}
The ratio of the reciprocal of (\ref{eqn:rrr}) to the QFI is given by
\begin{equation}
    {(\Delta \tilde{\theta})^{-2} \vert_{\theta=0}\over \mathcal{F}(m)}={\left(1-{2m^{2}\over {N^{2}\over 4} - m^{2}+{N\over 2}} \right)^{2}\over 4m^{2}+1}
    \label{eqn:utut}
\end{equation}
which satisfies (\ref{eqn:asas}). The last statement in the theorem follows from taking $m$ to satisfy the scaling assumption in the statement, and noting that (\ref{eqn:utut}) implies that
${1\over N}(\Delta \tilde{\theta})^{-2} \sim {1\over N(4m^{2}+1)}\left( {N^{2}\over 2} - 2m^{2}+N\right)$.
\end{proof}

 \begin{figure*}[t!]
    \centering
    \includegraphics[scale=1]{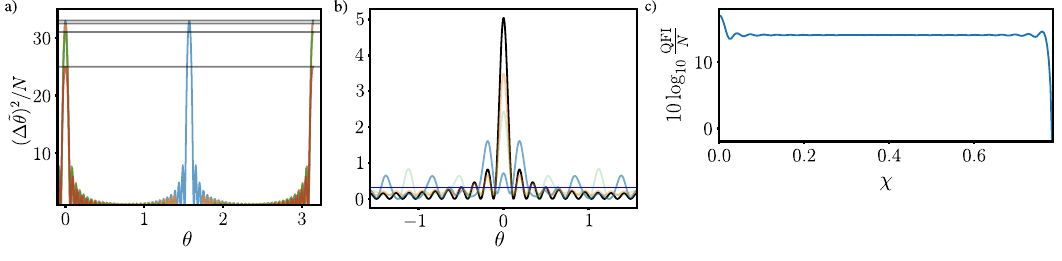}
    \caption{a) Error of the method of moments estimator $\tilde{\theta}$ of MZ angle $\theta$ with parity readout and with pure Dicke state probe ($N=64$; blue $m=0$, orange $m=4$, green $m=8$, red $m=16$). Black horizontal lines are the QFI. b) One-shot posterior probability density for the interferometric phase $\theta$ with the bosonic phase diffused probe state (\ref{eqn:iii}) ($N=40$; $\chi= 0 \text{ (black)}, 0.20, 0.39, 0.59, \pi/4 \text{ (dark blue)}$). Phase diffusion reduces the central peak visibility. c) The QFI for (\ref{eqn:iii}) exhibits greater than $O(N)$ scaling except for a small domain of $\chi$ around $\pi/4$.}
    \label{fig:f1}
\end{figure*}

Note that $\lim_{N\rightarrow \infty}N(\Delta \tilde{\theta})^{2}\vert_{\theta=0}=0$ if and only the reciprocal of the method of moments error is asymptotically greater than the SQL value $N$ for the QFI by more than a constant factor. We conclude that the $J_{z}^{2}$ measurement is useful for better-than-SQL interferometry with Dicke states of the form $\ket{{N\over 2}+o(\sqrt{N}),{N\over 2}-o(\sqrt{N})}$. Although a $J_{z}^{2}$ readout does not give a method of moments estimator that saturates the QFI, the correlation of this signal with the parity signal can be used to broaden the domain of $\theta$ where QFI saturation occurs, as shown Fig. \ref{fig:f1}a) (although global saturation still does not occur). Specifically, taking $\vec{O}=(O_{1},\ldots,O_{K})$ to be a row vector of $K$ observables, the inequality \cite{Gessner2020}
\begin{eqnarray}
    &{} {d\langle \vec{O}\rangle_{\rho_{\theta}}\over d\theta}\text{Cov}_{\rho_{\theta}}(\vec{O})^{-1}{d\langle \vec{O}\rangle_{\rho_{\theta}}^{\intercal} \over d\theta}\le \mathcal{F}(\rho_{\theta}) \\
    &{} \text{Cov}_{\rho_{\theta}}(\vec{O}):= \langle (\vec{O}-\langle \vec{O}\rangle_{\rho_{\theta}}) \circ (\vec{O}^{\intercal}-\langle \vec{O}^{\intercal}\rangle_{\rho_{\theta}}) \rangle_{\rho_{\theta}}\nonumber
    \label{eqn:bigmom}
\end{eqnarray}
holds for general parameterized states $\rho_{\theta}$, and one can consider $K=2$ with
$O_{1}=(-1)^{b^{\dagger}b}$, $O_{2}=J_{z}^{2}$. In (\ref{eqn:bigmom}), the symbol $\circ$ is the Jordan product of matrices $X\circ Y:={1\over 2}XY + {1\over 2}YX$. The inequality (\ref{eqn:bigmom}) generalizes the relationship between the signal-to-noise ratio in (\ref{eqn:methmom}) and the QFI, therefore the left hand side of (\ref{eqn:bigmom}) can be considered as a generalized signal-to-noise ratio.  By considering more quadratic spin observables in (\ref{eqn:bigmom}), it is possible to dispense with the parity readout and globally saturate (\ref{eqn:bigmom}) for all $\theta$. Taking $O_{1}=J_{z}^{2}$ and $O_{2}={1\over 2}(J_{+}^{2}+J_{-}^{2})$ if $m=0$, and taking $O_{1}=J_{z}^{2}$, $O_{2}={1\over 2}(J_{+}^{2}+J_{-}^{2})$, $O_{3}=J_{x}$ if $m\ge 1$, one obtains minimal operator lists for which (\ref{eqn:bigmom}) is globally satisfied. A similar conclusion is also obtained in the lossy TF setting in Section \ref{sec:lossytf}. Instead of providing individual proofs of these statements, we provide the full details for the QFI saturation in the case of gradiometry with doubled TF probes (Section \ref{sec:grad} and Appendix \ref{sec:appc}). The method of proof for the simpler statements of global saturation in this section and in Section \ref{sec:lossytf} follow straightforwardly. 

It is worth noting that the $N$-particle one-axis twisted probe state \cite{PhysRevA.47.5138}
\begin{equation}
    \ket{\psi_{\text{OAT}}(t)}:= e^{-it J_{z}^{2}}\ket{+}^{\otimes N}
    \label{eqn:tweltwel}
\end{equation}
at oversqueezed interaction times $t=O(N^{-\alpha})$, with $0<\alpha<1/2$, has asymptotic QFI equal to ${N(N+1)\over 2}$ for MZ interferometry, which is asymptotically equal to that of the TF probe state of the same particle number \cite{volkmart}. However, a twist-untwist protocol \cite{PhysRevLett.116.053601,PhysRevResearch.4.013236} (which is
obtained from Eq. (\ref{eqn:tweltwel}) by applying the inverse one-axis twisting $e^{+itJ_{z}^{2}}$ 
after the MZ interferometer) combined with a total spin readout only saturates the QFI near $\theta=0$, and the same is true for a parity and $J_{z}^{2}$ readout without the untwist operation. However, we are not aware of a quadratic spin readout that globally saturates the QFI for a one-axis twisted probe state in this interaction time domain.

Lastly, we note that a measurement in the $J_{y}$ phase basis 
\begin{equation}
    \ket{\tilde{k}}:={1\over \sqrt{N+1}}\sum_{n=0}^{N}e^{2\pi i n k \over N+1}e^{i{\pi\over 2}J_{x}}\ket{N-n,n} \, , \, k=0,\ldots,N
    \label{eqn:ps}
\end{equation}
(eigenvectors of $e^{i{\pi\over 2}J_{x}}Ce^{-i{\pi\over 2}J_{x}}$ where $C$ is the cyclic shift $C\ket{N-n,n}=\ket{N-n-1,n+1}$ with addition modulo $N+1$)
contains sufficient information to form a globally optimal estimator of $\theta$. This can be shown by numerical computation of the classical Fisher information with respect to measurement of the complete, orthonormal basis (\ref{eqn:ps}) and verifying that it saturates the black horizontal lines in Fig. \ref{fig:f1}a) for all $\theta$.

\section{Phase-diffused and lossy TF probes\label{sec:lossytf}}

Interatomic interactions during the interferometric sequence lead to the phenomenon of bosonic phase diffusion \cite{PhysRevLett.78.4675,PhysRevLett.75.2944,PhysRevLett.98.030407}. Unlike optical phase diffusion, which is modeled by applying random unitary optical rotations  to a continuous-variable probe state \cite{Vidrighin2014,PhysRevLett.106.153603}, bosonic phase diffusion is a unitary error which reduces the phase coherence in the arms of the interferometer. In a modern context, the fact that one-axis twisting can reduce the performance of atom interferometry may seem surprising given that the one axis-twisted split Bose-Einstein condensate $e^{-i\chi J_{z}^{2}}e^{-i{\pi\over 2}J_{y}}\ket{N,0}$ exhibits Heisenberg scaling as a probe for MZ interferometric phase (i.e., $J_{y}$ rotation) for $O(N^{-1/2})\lesssim \chi \lesssim o(1)$, transitioning to Heisenberg limit scaling at $\chi \sim O(1)$ \cite{volkmart}. To demonstrate the effect of bosonic phase diffusion on interferometric sensitivity, we utilize a well-known Bayesian parameter estimation method \cite{PhysRevLett.71.1355,PhysRevA.73.011801,PhysRevA.76.013804} and the phase-diffused probe state
\begin{equation}
    \ket{\psi_{N}(\chi)}:= e^{-i\chi J_{z}^{2}}e^{-i{\pi\over 2}J_{y}}\ket{{N\over 2},{N\over 2}}.
    \label{eqn:iii}
\end{equation}
Applying the MZ interferometer to (\ref{eqn:iii}) and using the fidelity with $e^{-i{\pi\over 2}J_{y}}\ket{{N\over 2},{N\over 2}}$ to define a signal, one  obtains the posterior probability densities in Fig. \ref{fig:f1} for various $\chi$ (a uniform prior for the phase $\theta$ is taken on $[-\pi/2,\pi/2]$). At $\chi=\pi/4$, sensitivity to the MZ phase disappears because the probe state is equal to $e^{i{\pi\over 2}J_{x}}\ket{{N\over 2},{N\over 2}}$, which is the $0$ eigenvector of $J_{y}$. Because of the robustness of the Heisenberg scaling of the TF state to strong bosonic phase diffusion, consistent with a more detailed analysis of the effect of phase diffusion on the sensitivity of the local density readout for MZ interferometry with number-squeezed bosons in a double-well trap \cite{Grond_2010}, we do not further consider its  effect on the probe states.

We now analyze the QFI after loss of $K$ atoms from the TF state, with the aim of identifying an asymptotic scaling of $K$ with $N$ that allows better-than-SQL scaling to persist. Analyses of mixtures of Dicke states with uncertain total number of particles appear in Ref.\cite{Meiser_2009} and with uncertain spin projection in Ref.\cite{Apellaniz_2015}. Performance of the $J_{z}^{2}$ measurement when the TF state is well-formed, but the total number of particles of the TF state has a non-sharp distribution, was analyzed in Ref.\cite{tmw}. We call the quantum channel describing partial trace over $K$ particles by $\mathcal{E}_{K}$.
Starting with a Dicke state $\ket{N-m,m}$, the quantum state subsequent to loss of $K<N-m$ atoms can be computed from the following time-inhomogeneous Markov process (for proof, see Appendix \ref{sec:appb})
\begin{equation}
    p_{k+1}=Q_{k}p_{k} \; , \; k=0,\ldots, K-1
    \label{eqn:fff}
\end{equation}
where $p_{0}=(0,\ldots,0,1)^{\intercal}$ is a vector of length $m+1$,  $Q_{k}$ is a sequence of $(m+1)\times (m+1)$ bidiagonal, bistochastic matrix defined by nonzero elements
\begin{align}
    (Q_{k})_{i,i+1}&={i\over n_{k}} \nonumber \\
    (Q_{k})_{i,i}&=1-{(i-1)\over n_{k}},
    \label{eqn:qmat}
\end{align}
for $i=1,\ldots, m$, and $n_{k}=N-k$ is the particle number after loss of $k$ particles, $k=0,\ldots, K$.
The dynamics occurs on the space of probability distributions on the discrete set $\lbrace 0,\ldots, m\rbrace$ due to the fact that loss of an atom cannot increase the Hamming weight that defines a Dicke state.
To apply (\ref{eqn:fff}) to the TF state, we take $m=N/2$ and numerically compute $p_{K}$ and note that the final state of $N-K$ particles $\rho_{N,K}:=\mathcal{E}_{K}\left( \ket{{N\over 2},{N\over 2}}\bra{{N\over 2},{N\over 2}}\right)$  is a statistical mixture of Dicke states given by
\begin{equation}
\rho_{N,K}:= \sum_{i=1}^{N/2 +1}(p_{K})_{i}P^{{N-K\over 2}}_{i-1}
\label{eqn:rhonk}
\end{equation}
where $P^{J}_{\ell}$ is the projection to Dicke state $\ket{2J-\ell,\ell}$. 
The probe state for MZ interferometry is then obtained from Eq.(\ref{eqn:rhonk}) by taking $\rho_{N,K}(\theta):=e^{-i\theta J_{y}}\rho_{N,K}e^{i\theta J_{y}}$. It should be noted that the spin projection is invariant under the particle loss, i.e., $
\langle J_{z}\rangle_{\rho_{N,K}}=0$. The QFI $\mathcal{F}(\rho_{N,K})$  for this probe state can be computed from the spectral formula \cite{PhysRevLett.72.3439}, using the state $(1-\epsilon)\rho_{N,K}(\theta) + {\epsilon \over N-K+1}\mathbb{I}_{N-K+1}$ for an infinitesimal $\epsilon$, due to the fact that $\rho_{N,K}(\theta)$ is not generally full rank. The decrease in the QFI with respect to increasing $K$ in shown in Fig. \ref{fig:f2}. The loss of one particle decreases the QFI by a factor of 2 asymptotically, with the exact value given by $\left({N\over 2}\right)^{2}-1$, which is proven in Appendix \ref{sec:appb}. This behavior can be contrasted with the optimal pure state probe (viz., the GHZ state), which exhibits SQL scaling if even one particle is lost. In this work, we assume that the loss $K$ is known. To describe an experiment, it would be appropriate to consider a convex mixture of $\mathcal{E}_{K}$ channels to describe the probabilistic loss of particles.

Unlike the case for the noiseless TF probe in Section \ref{sec:tt}, it is not possible to saturate the QFI $\mathcal{F}(\rho_{N,K})$ for general loss values $K$ using state-agnostic readouts such as parity or $J_{z}^{2}$. For a fixed loss value $K$, the optimal readout  depends on the spectrum of the noisy state, i.e., on the  $(p_{K})_{m}$, through the symmetric logarithmic derivative \cite{holevobook}. 
\begin{figure*}[t!]
    \centering
    \includegraphics[scale=1]{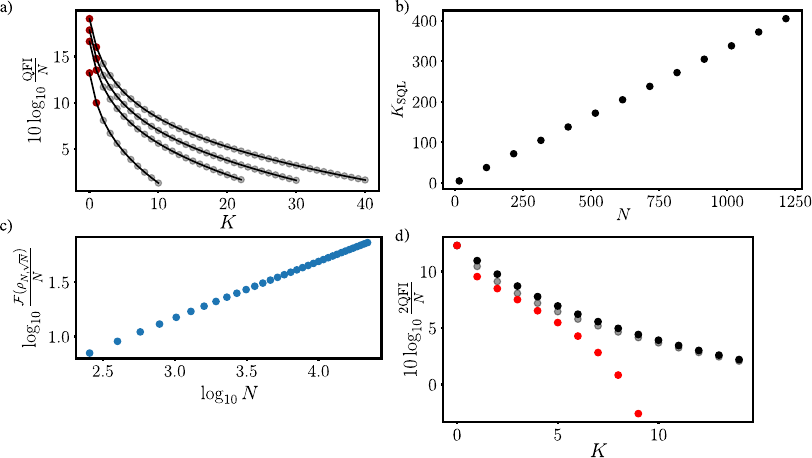}
    \caption{a) QFI after application of $\mathcal{E}_{K}$ for $K=0,\ldots,\lfloor {N\over 4}\rfloor$ with $N=40$ (bottom curve), 90, 120, 160 (top curve). $K=0$ analytical value ${N^{2}\over 2}+N$ and $K=1$ analytical value ${(N/2)^{2}-1}$ are shown as red dots. b) Maximal loss tolerance of SQL scaling of the QFI for $N=16,\ldots,1216$ in steps of $100$. c) $\mathcal{F}(\rho_{N,\sqrt{N}})$ scaling as $O(N^{1.52})$ for large $N$. d) (black dots) $(1,1)$ element of QFI matrix after loss of $K=0,1,\ldots,14$ particles from probe state (\ref{eqn:iid}) with $N=64$; (gray dots) method of moments error for local $J_{x}$, $J_{z}^{2}$,  ${1\over 2}(J_{+}^{2}+h.c.)$, and ${1\over 2}(J_{+}J_{z}+h.c.)$ readouts; (red dots) QFI after application of $\mathcal{E}_{K}$ to the TF state with $N=32$. The $K=0$ values coincide at $10\log_{10}17$.
    \label{fig:f2}}
\end{figure*}
It can be expected that the method of moments error for the parity readout saturates $\mathcal{F}(\rho_{N,K})$ near $\theta=0$, since this state is a statistical mixture of Dicke states which, if any of these is used as a probe state, allows to saturate $\mathcal{F}(m)$ as in Fig. \ref{fig:f1} a. However, applying the particle loss channel $\mathcal{E}_{K}$ results in a state with low parity signal in the interferometer. For instance, applying $\mathcal{E}_{1}$ to the TF state gives an equally weighted statistical mixture of Dicke states with different parity, and after applying the MZ interferometer for small $\theta$ the state still has equal weights in the even and odd parity sectors. Even taking into account the $J_{z}^{2}$ readout and its correlation with the parity readout via (\ref{eqn:bigmom}) to get a smaller error, saturation of the QFI does not occur for any $\theta$ domain for loss values $K>1$. In fact, by comparing the QFI $\mathcal{F}(\rho_{N,K})$ to the classical Fisher information $F_{K}(\theta)=\sum_{\ell=0}^{N/2}q_{\theta}(\ell)^{-1}\left( {dq_{\theta}(\ell)\over d\theta}\right)^{2}$ where
\begin{equation}
    q_{\theta}(\ell)=\langle N-K-\ell,\ell\vert \rho_{N,K}(\theta)\vert N-K-\ell,\ell\rangle
    \label{eqn:cfiz}
\end{equation}
is the probability of observing $\ell$ particles in the second mode, one finds that a particle number-resolving measurement does not provide enough information to construct an estimator of $\theta$ with error saturating the inverse QFI for $K>1$. Similar to Ref.\cite{Apellaniz_2015}, we find that the greatest sensitivity is obtained away from $\theta=0$, even for a combined $J_{z}^{2}$ and parity readout. Instead of a measurement in the basis of $J_{z}$ eigenvectors as in Eq. (\ref{eqn:cfiz}), one could also consider the classical Fisher information of a measurement in the $J_{y}$ phase basis Eq. (\ref{eqn:ps}). However, the classical Fisher information of this measurement is found not to saturate $\mathcal{F}(\rho_{N,K})$ for $K>0$, but rather has a periodic structure which causes the greatest sensitivity to occur at a set of equally-spaced points in $(-\pi,\pi]$. 

These results indicate that neither a measurement in the occupation number basis nor a measurement in the $J_{y}$ phase state basis produce enough information to saturate the QFI for a lossy twin Fock probe. However, we have verified numerically that there indeed exists a minimal list $\vec{O}$ of observables such that the method of moments error (left hand side of (\ref{eqn:bigmom})) saturates the QFI $\mathcal{F}(\rho_{N,K})$ for all $K$ and globally for all $\theta$. Specifically, for loss values $K=0,1$, the list 
\begin{equation}
\vec{O}=\lbrace J_{z}^{2},{1\over 2}J_{+}^{2}+h.c.\rbrace
\end{equation} is sufficient, and for $K=2,\ldots,{N\over 2}$, appending the linear spin observable $J_{x}$ and the quadratic spin observable ${1\over 2}J_{z}J_{+}+h.c.$ to the list $\vec{O}$, forming 
\begin{equation}
\vec{O}= \lbrace J_{x}, J_{z}^{2},{1\over 2}J_{+}^{2}+h.c.,{1\over 2}J_{z}J_{+}+h.c.\rbrace,
\end{equation} is sufficient. Remarkably, a phase observable with eigenvectors in Eq. (\ref{eqn:ps}), or the nonlinear parity observable $(-1)^{b^{\dagger}b}$ are not necessary for globally optimal interferometry with lossy TF states.

From the discussion in Section \ref{sec:tt}, we recall that a sequence of  corrupted (but pure) TF states $\ket{{N\over 2}+\lambda N,{N\over 2}-\lambda N}$ with $\lambda <1/2$ maintain Heisenberg scaling of the QFI. To understand the  tolerance of the lossy probe state $\rho_{N,K}$, we found for each total particle number value $N$, the greatest loss $K=K_{\text{SQL}}(N)$ such that the $\mathcal{F}(\rho_{N,K_{\text{SQL}}(N)})>N$. Such a loss value is interpreted as the maximal number of particles that may be lost while maintaining better-than-SQL value for the QFI. Fig. \ref{fig:f2}b indicates that $K_{\text{SQL}}(N)$ scales linearly with $N$, and curve fitting to an affine function suggests that 1/3 of the particles can be lost while maintaining better-than-SQL scaling. It is also of interest to identify exponents $0<\alpha<1$ and $\beta>1$ for which $\mathcal{F}(\rho_{N,N^{\alpha}}) = O(N^{\beta})$. These exponents characterize the amount of loss (as a function of the total particle number) that can be tolerated while maintaining a QFI scaling as $N^{\beta}$. For $\beta=2$, i.e., Heisenberg scaling, we were not able to identify the existence of an $\alpha$. However, the loss exponent $\alpha = 1/2$ allows QFI scaling of $N^{\beta}$ with $\beta \approx 3/2$, as shown in Fig. \ref{fig:f2}c.


\section{Gradiometry with doubled TF and states\label{sec:grad}}
Advantages of using pure, spatially-split spin squeezed atomic ensembles for magnetic gradiometry (estimation of the difference $\theta_{1}-\theta_{2}$) beyond SQL were outlined in Ref.\cite{Fadel_2023}. Generally, when employing atom interferometers for distributed sensing (i.e., sensing a field at many different spatial points), there are two classes of quantum sensing strategies that can be employed, which are straightforward generalizations of their single-parameter sensing counterparts \cite{PhysRevLett.96.010401}. The \textit{parallel strategies} utilize identical copies of a probe state or a single global entangled state as a probe which addresses the various points of interest. The probe state is locally parametrized by phases corresponding to the local field values, with the parametrization usually modeled by a tensor product of parametrized quantum channels. By contrast, the \textit{sequential strategies} expose an initial probe state (which may be entangled with an ancilla register) to the points of interest in temporal succession. The preference of strategy depends on the cost of probe state preparation, the dwell time of the transient that must be sensed, etc. It has been shown that in the parallel generalization of the MZ interferometry protocol which utilizes general linear optical unitaries in lieu of the $e^{\pm i{\pi\over 2}J_{x}}$ two-mode beamsplitters, twin Fock states asymptotically allow to saturate the optimal performance (over all mode-separable input states of fixed total  particle number) for estimation of a linear function of phases \cite{PhysRevLett.121.043604}.  Often in distributed sensing, only a single linear function $\vec{w}\cdot \vec{\theta}$ is of the parameters is of interest \cite{PhysRevLett.120.080501}, with the normalization $\Vert \vec{w}\Vert_{2}=1$ sometimes chosen so that the largest possible QFI for estimating a linear function  coincides with the largest eigenvalue of the QFI matrix.

One possible sequential strategy to estimate $(\theta_{1}-\theta_{2})/\sqrt{2}$ using a TF or Dicke state probe is to $\pi$ shift the phase of the second MZ interferometer addressing the $\theta_{2}$ parameter. This results in a full protocol described by the operation
\begin{equation}
    e^{-i{\pi\over 2}J_{x}}e^{-i\theta_{2}J_{z}}e^{i{\pi\over 2}J_{x}}e^{i{\pi\over 2}J_{x}}e^{-i\theta_{1}J_{z}}e^{-i{\pi\over 2}J_{x}} = e^{-i(\theta_{1}-\theta_{2})J_{y}}
    \label{eqn:jhjh}
\end{equation}
However, note that for any global rotation $e^{-i\varphi \vec{n}\cdot \vec{J}}$, with $\vec{n}\in \mathbb{R}^{3}$ a unit vector and $\vec{J}=(J_{x},J_{y},J_{z})$ the vector of spin operators, the following channel commutativity relation holds on the quantum states of $N$ two-mode bosons
\begin{equation}
    \mathcal{E}_{K}\circ e^{-i\varphi \vec{n}\cdot \vec{J}} = e^{-i\varphi \vec{n}\cdot \vec{J}} \circ \mathcal{E}_{K}
\end{equation}
where the spin operators act in a spin-$N/2$ representation on the left hand side, and a spin-${N-K\over 2}$ representation on the right hand side. Therefore, an analysis of this sequential protocol (\ref{eqn:jhjh}) can be carried out (even with imperfect probes or particle loss) using the methods of Sections \ref{sec:tt} and \ref{sec:lossytf}. 

Instead, we consider two parallel strategies in which globally entangled states are used to estimate linear functions of the phases $\theta_{j}$, $j=1,2$. The first parallel strategy is especially relevant when the initial TF state is a single atomic cloud with $N/2$ occupation in each of two internal modes, e.g., atomic nuclear spin states. The initial TF cloud is then spatially split (we assume perfect splitting) to give the TF probe state in the superposition of spatial modes
\begin{equation}
    \ket{\psi_{\text{split}}}:=\left( {a_{1}^{\dagger}+a_{2}^{\dagger}\over \sqrt{2}}\right)^{N/2}\left( {b_{1}^{\dagger}+b_{2}^{\dagger}\over \sqrt{2}}\right)^{N/2}\ket{0,0,0,0}.
    \label{eqn:split}
\end{equation}
Note that the original internal modes $a$ and $b$ of the TF state have been spatially split to spatial modes $a_{1}$, $a_{2}$ and $b_{1}$, $b_{2}$, respectively. The local phase shifts are imprinted on $\ket{\psi_{\text{split}}}$ by the operation $U_{\text{MZ}}(\theta)=e^{-{\theta_{1}\over 2}(a_{1}^{\dagger}b_{1}-h.c.)-{\theta_{2}\over 2}(a_{2}^{\dagger}b_{2}-h.c.)}$. This strategy appears especially relevant for magnetic gradiometry with atomic ensembles. However, we find that every entry of the inverse of the QFI matrix scales as $N^{-1}$. The multiparameter quantum Cram\'{e}r-Rao bound then implies that an estimator of any linear function of $\theta_{1}$ and $\theta_{2}$ has an error scaling at least as $N^{-1}$, i.e., scaling as SQL. No improvement can be obtained by losing particles, so we cease further exploration of probe state $\ket{\psi_{\text{split}}}$.

The second parallel strategy uses the same local phase shift operation $U_{\text{MZ}}(\theta)$, but utilizes as probe the doubled TF state
\begin{equation}
\ket{\psi_{\text{doubled}}}:= a_{1}^{\dagger {N\over 4}}b_{1}^{\dagger {N\over 4}}a_{2}^{\dagger {N\over 4}}b_{2}^{\dagger {N\over 4}}\ket{0,0,0,0}
\label{eqn:iid}
\end{equation}
in which two spin modes and two spatial modes are occupied symmetrically. The state (\ref{eqn:iid}) can also be considered as two copies of \textit{bosonic independent and identically distributed} (b.i.i.d.) TF states \cite{volkbiid} which, mathematically, corresponds to taking the Young product of two copies of the symmetric subspace of $(\mathbb{C}^{2})^{\otimes N}$. Note that both (\ref{eqn:split}) and (\ref{eqn:iid}) are in the symmetric subspace of $(\mathbb{C}^{4})^{\otimes N}$ and are, therefore, valid bosonic states of $N$ atoms. For independent MZ interferometers acting on the $a_{j},b_{j}$ mode pairs, the $2\times 2$ QFI matrix for (\ref{eqn:iid}) is diagonal with equal $(1,1)$ and $(2,2)$ matrix elements given by ${N^{2}+4N\over 8}$, independent of $\vec{\theta}$ (the probe state (\ref{eqn:split}) does not have a diagonal QFI matrix). Therefore, an optimal estimator of any linear function $\vec{w}\cdot \vec{\theta}$, $\Vert \vec{w}\Vert_{2}=1$ exhibits Heisenberg scaling, and we will therefore quantify the performance of a probe state for gradiometry by the $(1,1)$ entry of QFI matrix. Using (\ref{eqn:bigmom}), one can show analytically that, e.g., the $(j,j)$ element of the QFI matrix is globally saturated by method of moments estimation of the local quadratic spin observables $O_{1}=J_{z}^{(a_{j},b_{j}) 2}$, $O_{2}={1\over 2}(J_{+}^{(a_{j},b_{j})2}+h.c.)$. The covariance matrix and derivatives are shown in Appendix \ref{sec:appc}, along with  verification that the final expression indeed simplifies to ${N\over 8}(N+4)$ regardless of $\theta_{1}$. 

The property of the QFI matrix being a scalar multiple of the identity holds also for the probe state obtained by losing $K$ particles from (\ref{eqn:iid}). This fact follows because taking the trace over $K$ particles produces a statistical mixture of four-mode Dicke states and the $1\leftrightarrow 2$ label symmetry of the state is preserved.  Note that atom loss from (\ref{eqn:iid}) is not a strictly local process to the MZ interferometers and cannot be desribed by, e.g., applying the loss channels $\mathcal{E}_{K}$ to the mode pairs $a_{1},b_{1}$ and $a_{2},b_{2}$ independently. Further, although the dimension of the symmetric subspace of $(\mathbb{C}^{4})^{\otimes N}$ scales as $N^{3}$, if one considers losing at least $K$ atoms from the doubled TF state with $K\ge {N\over 4}$, one must keep track of at least $e^{O(\sqrt{N})}$ weights for an exact description of the state, according to the Hardy-Ramanujan asymptotic for the number of partitions of a large natural number. We briefly describe the channel describing particle loss from a Dicke state with more than two modes in Appendix \ref{sec:appb}.

The black dots of Fig. \ref{fig:f2}d) show the gain over the SQL ${N\over 2}$ for estimation of the first MZ phase $\theta_{1}$ (with $N=64$). The SQL is taken as ${N\over 2}$ because at most $N/2$ particles address each of the MZ interferometers when loss is applied to probe state (\ref{eqn:iid}). The red dots show the gain over SQL for the probe state $\mathcal{E}_{K}\left( \ket{{N\over 4},{N\over 4}}\bra{{N\over 4},{N\over 4}}   \right)$ using the same analysis as in Section \ref{sec:lossytf}. It is clear that for small loss values, the four mode state exhibits roughly 1.4 dB gain in the attainable precision of local interferometry compared to a single lossy TF state with the same local energy and loss. The result indicates a different phenomenon from the known distributed sensing results in the multimode optical setting with pure probe states \cite{PhysRevLett.121.043604} or massive boson setting with pure probe states \cite{PhysRevLett.121.130503}. In the absence of loss, the optimal strategy for estimation of a linear function of the MZ phases $\theta_{1}$, $\theta_{2}$ using the probe state (\ref{eqn:iid}) does not outperform the use of two independent TF states $\ket{{N\over 4},{N\over 4}}$ to separately probe the MZ interferometers, which indeed indicates that the probe state (\ref{eqn:iid}) is a suboptimal four mode, $N$ particle probe state (one could apply a linear optical operation to, e.g., $\ket{{N\over 2},{N\over 2},0,0}$ to obtain a better probe state) \cite{PhysRevLett.121.043604}. However, when at least one particle is lost, the distributed entanglement of the lossy version of (\ref{eqn:iid}) allows greater attainable precision than independent copies of the lossy TF probe state.

\section{Discussion\label{sec:disc}}
We analyzed method of moments readouts for the TF state and its images under particle loss channel or unitary phase-diffusion channel. The traditional parity measurement, which saturates the QFI for MZ interferometry near $\theta=0$ for all Dicke states, can be dispensed with if one can measure both quadratic spin observables $J_{z}^{2}$ and ${1\over 2}(J_{+}^{2}+h.c.)$. Although one-shot measurements of both these observables is impossible because they don't commute, allocation of shots to one or the other observable allows to extract an estimator with error asymptotically given by the left hand side (\ref{eqn:bigmom}), according to the central limit theorem. When particle loss is taken into account, we identified minimal lists of observables that give method of moments error that globally saturates the QFI for any loss value. The observables are again at-most-quadratic 
in the spin operators. The fact that better than SQL scaling is obtained even when $O(\sqrt{N})$ particles are lost from the TF state indicates the loss robustness of the TF probe state for practical MZ interferometry.
To analyze a formal gradiometry protocol, we considered application of particle loss to the doubled twin Fock state $\ket{\psi_{\text{doubled}}}=\ket{{N\over 4},{N\over 4},{N\over 4},{N\over 4}}$, which can probe two MZ interferometers describing, e.g., distributed sensing of spatially separated external field values. Although the noiseless doubled TF state does not allow to estimate $\theta_{1}$ (chosen without loss of generality) with lower error than is achievable by probing the MZ interferometer with a TF states of $N/2$ particles, the lossy doubled TF state exhibits an advantage over the lossy TF state. This result suggests that an extended bosonic insulating state is, in a practical lossy setting, a more useful resource for distributed quantum sensing than a tensor product of states of fixed particle number. The advantage can be attributed to the indistinguishability of the particles in the bosonic insulator, giving a fully symmetrized state which causes the loss to be distributed over all occupied modes. Extending this distributed sensing result to include other noise sources e.g., thermal \cite{PhysRevA.98.032325}, would establish this noisy advantage in experimentally realistic settings such as those realized in recent demonstrations of entanglement between spatially separated atom ensembles suggests  \cite{tothapell}.

\acknowledgements
The authors thank Katarzyna Krzyzanowska, Sivaprasad Omanakuttan, and Jonathan Gross for helpful discussions and acknowledge support from the Laboratory Directed Research and Development (LDRD) Program at Los Alamos National Laboratory (LANL).
 Los Alamos National Laboratory is managed by Triad National Security, LLC, for the National Nuclear Security Administration of the U.S. Department of Energy under Contract No. 89233218CNA000001.
\onecolumngrid
\bibliography{test.bib}

\onecolumngrid
\appendix
\section{\label{sec:appa}Moments of $J_{z}^{2}$}
We show the first and second moments of $J_{z}^{2}$ utilized in the calculation of the method of moments error in Theorem 1.
\begin{align}
    \langle J_{z}^{2}\rangle_{e^{i\theta J_{y}}\ket{{N\over 2}-k,{N\over 2}+k}}&= k^{2}\cos^{2}\theta + {\sin^{2}\theta\over 2}\left( {N^{2}\over 4} - k^{2} + {N\over 2}\right) \nonumber \\
    {d\over d\theta}\langle J_{z}^{2}\rangle_{e^{i\theta J_{y}}\ket{{N\over 2}-k,{N\over 2}+k}}&= {\sin2\theta\over 2}\left( {N^{2}\over 4} - 3k^{2} + {N\over 2}\right)\nonumber \\
    \langle J_{z}^{4}\rangle_{e^{i\theta J_{y}}\ket{{N\over 2}-k,{N\over 2}+k}}&= \left( k^{2}\cos^{2}\theta + {\sin^{2}\theta\over 2}\left( {N^{2}\over 4}-k^{2}+{N\over 2}\right) \right)^{2} \nonumber \\
    &+ {(2k-1)^{2}\sin^{2}\theta\cos^{2}\theta \over 4}\left( {N\over 2}-k+1 \right) \left( {N\over 2}+k \right) \nonumber \\
    &+ {(2k+1)^{2}\sin^{2}\theta\cos^{2}\theta \over 4}\left( {N\over 2}+k+1 \right) \left( {N\over 2}-k \right) \nonumber \\
    &+ {\sin^{4}\theta \over 16} \left( {N\over 2}-k+1\right)\left({N\over 2}-k+2 \right)\left( {N\over 2}+k\right)\left( {N\over 2}+k-1\right)\nonumber \\
    &+ {\sin^{4}\theta \over 16} \left( {N\over 2}-k\right)\left({N\over 2}-k-1 \right)\left( {N\over 2}+k+1\right)\left( {N\over 2}+k+2\right)
    \label{eqn:uuu}
\end{align}

\section{\label{sec:appb} Proof of Eq. (\ref{eqn:rhonk}) and QFI of $\mathcal{E}_{1}(\ket{{N\over 2},{N\over 2}}\bra{{N\over 2},{N\over 2}})$}

Consider a Dicke state $\ket{N-m,m}$ with $m\le {N\over 2}$, which we associate with the vector $p_{0}=(0,\ldots,0,1)^{\intercal}$ of length $m+1$ by writing
\begin{equation}
    \ket{N-m,m}\bra{N-m,m} = \sum_{\ell=1}^{m+1}(p_{0})_{\ell}\ket{N-(\ell-1),\ell-1}\bra{N-(\ell-1),\ell-1}.
\end{equation}
Define the sequence $n_{k}=N-k$ for $k=0,\ldots,K$ as the total number of particles present after loss of $k$ particles.
Using the ``first quantized'' expression analogous to (2), one finds that tracing over the $N$-th particle (choosing the $N$-th particle without loss of generality because of the permutation symmetry of the state) results in the state
\begin{align}
    &{} \left(1-{m\over n_{0}}\right) \ket{N-1-m,m}\bra{N-1-m,m} + {m\over n_{0}}\ket{N-m,m-1}\bra{N-m,m-1} \nonumber \\
    &{} = \sum_{\ell=1}^{m+1}(Q_{0}p_{0})_{\ell}\ket{N-1-(\ell-1),\ell-1}\bra{N-1-(\ell-1),\ell-1}\nonumber \\
    &{} =  \sum_{\ell=1}^{m+1}(Q_{0}p_{0})_{\ell}P^{{N-1\over 2}}_{\ell-1}\nonumber \\
    &{}= \sum_{\ell=1}^{m+1}(p_{1})_{\ell}P^{{N-1\over 2}}_{\ell-1}
    \label{eqn:bbt}
\end{align}
where $Q_{k}$ is the sequence of matrices in (16) and $P^{J}_{\ell}$ is the projection to Dicke state $\ket{2J-\ell,\ell}$ for any spin-$J$ SU(2) representation. Note that one can interpret the above state as the dynamics corresponding to a particle having been lost from the $\ket{0}$ mode ($\ket{1}$ mode) with probability $1-{m\over n_{0}}$ ($m\over n_{0}$). Therefore, for general $k$, the $m$-th diagonal (superdiagonal) element of $Q_{k}$ corresponds to the probability of losing a particle from the $\ket{0}$ mode ($\ket{1}$ mode). Iterating this stochastic process up to the $K$-th loss gives the statistical mixture
\begin{equation}
    \sum_{\ell=1}^{m+1}(\prod_{k=0}^{K}Q_{k}p_{0})_{\ell}P^{{N-K\over 2}}_{\ell-1} = \sum_{\ell=1}^{m+1}(p_{K})_{\ell}P^{{N-K\over 2}}_{\ell-1} .
\end{equation}

We now calculate the QFI of the probe state
\begin{equation}
    e^{-i\theta J_{y}}\rho_{N,1}e^{i\theta J_{y}} = e^{-i\theta J_{y}}\mathcal{E}_{1} \left( \ket{{N\over 2},{N\over 2}}\bra{{N\over 2},{N\over 2}}\right)e^{i\theta J_{y}}
\end{equation}
analytically, where $\rho_{N,K}$ is given in (\ref{eqn:rhonk}). From (\ref{eqn:bbt}) one finds that $\rho_{N,1}$ is not full rank, so we replace the probe state by the $\epsilon$-parametrized full rank state
\begin{equation}
    (1-\epsilon)e^{-i\theta J_{y}}\rho_{N,1}e^{i\theta J_{y}} + \epsilon {\mathbb{I}_{N}\over N} = \sum_{j=0}^{N-1}\lambda_{j}(\epsilon)\ket{\psi_{j}(\epsilon)}\bra{\psi_{j}(\epsilon)}
    \label{eqn:bbb}
\end{equation}
where $N$ is the dimension of the symmetric subspace after loss of a particle. The spectral formula for the QFI is then given by
\begin{equation}
    \lim_{\epsilon\rightarrow 0}2\sum_{j,k}{(\lambda_{j}(\epsilon)-\lambda_{k}(\epsilon))^{2}\over \lambda_{j}(\epsilon)+\lambda_{k}(\epsilon)} \vert \langle \psi_{j}(\epsilon)\vert J_{y}\vert \psi_{k}(\epsilon) \rangle \vert^{2}
\end{equation}
which is computed analytically for $K\le 1$ and numerically for $K>1$ to give the data in Fig. 2a. For example, when $K=1$, two eigenvalues of (\ref{eqn:bbb}) are ${1-\epsilon \over 2}+{\epsilon\over N}$ and the remaining eigenvalues are ${\epsilon\over N}$. Computing the contributing Dicke state matrix elements and taking $\epsilon\rightarrow 0$ gives $\left( {N\over 2}\right)^{2}-1$, as stated in the main text.

To define the loss channel acting on an $L$ mode bosonic state of fixed particle number, it is convenient to consider sparse rank-$L$ tensors. Take $\vec{e}_{1}=(1,0,0,0)$, $\ldots$, $\vec{e}_{4}:=(0,0,0,1)$, and define the elements of the four-index tensor $Q$ by
\begin{equation}
    Q_{\vec{\ell}-\vec{e}_{j},\vec{\ell}}={\ell_{j}\over \Vert \vec{\ell}\Vert_{1}}. 
    \end{equation}
    Because tracing over one of the particles of $\ket{{N\over 4},{N\over 4},{N\over 4},{N\over 4}}$ reduces $\Vert \vec{\ell}\Vert_{1}$ by one, one obtains the statistical weights of the Dicke states after loss of $K$ particles as
    \begin{equation}
        p^{(K)}:=Q^{K}p^{(0)}
    \end{equation}
    where $p^{(0)}$ is a probability vector with unit weight on the vector index $\vec{\ell}=({N\over 4},{N\over 4},{N\over 4},{N\over 4})$.

\section{\label{sec:appc}Global saturation of $(1,1)$ element of QFI matrix for distributed interferometry}

As in Section 4 of the main text, we consider the distributed probe $\ket{\psi_{\text{doubled}}}:= \ket{{N\over 4},{N\over 4},{N\over 4},{N\over 4}}$ parametrized by the local phases $\theta_{1},\theta_{2}$ by
\begin{equation}
    U_{\vec{\theta}}:= e^{-i \theta_{1}J_{y}^{(a_{1},b_{1})} - i \theta_{2}J_{y}^{(a_{2},b_{2})}}.
\end{equation}
We consider the row vector of observables $\vec{O}=(O_{1},O_{2},O_{3},O_{4})$ where\begin{align}
    O_{1}&=J_{z}^{(a_{1},b_{1})2}\nonumber \\
    O_{2}&={1\over 2}(J_{+}^{(a_{1},b_{1})2}+h.c.)\nonumber \\
    O_{3}&=J_{z}^{(a_{2},b_{2})2} \nonumber \\
    O_{4}&={1\over 2}(J_{+}^{(a_{2},b_{2})2}+h.c.).
\end{align}Because there are two parameters instead of one, the inequality (11) is not immediately applicable. Instead, we use the matrix inequality
\begin{equation}
    M(\vec{\theta})\preceq \mathcal{F}(\vec{\theta})
\end{equation}
where the right hand side is the $2\times 2$ QFI matrix based on symmetric logarithmic derivatives and the matrix elements of $M(\vec{\theta})$ are
\begin{equation}
  M(\vec{\theta})_{i,j}=\del_{\theta_{i}}\langle \vec{O}\rangle_{\rho_{\vec{\theta}}}  \left( \text{Cov}_{\rho_{\vec{\theta}}}(\vec{O})\right)^{-1}\del_{\theta_{j}}\langle \vec{O}\rangle_{\rho_{\vec{\theta}}}^{\intercal} 
\end{equation}
with $\rho_{\vec{\theta}}:= U_{\vec{\theta}}\ket{\psi_{\text{doubled}}}\bra{\psi_{\text{doubled}}}U_{-\vec{\theta}}$ and $(\text{Cov}_{\rho_{\vec{\theta}}}(\vec{O}))_{i,j}={1\over 2}\langle ( O_{i}-\langle O_{i}\rangle_{\rho_{\vec{\theta}}})\circ (O_{j}-\langle O_{j}\rangle_{\rho_{\vec{\theta}}})\rangle_{\rho_{\vec{\theta}}}$ are the elements of the $4\times 4$ covariance matrix. However, $\langle O_{i}O_{i'} \rangle_{\rho_{\vec{\theta}}}= \langle O_{i} \rangle_{\rho_{\vec{\theta}}}\langle O_{i'} \rangle_{\rho_{\vec{\theta}}}$ for $i\in \lbrace 1,2\rbrace$, $i'\in \lbrace 3,4\rbrace$, so $\text{Cov}_{\rho_{\vec{\theta}}}(\vec{O})$ breaks into a direct sum of $2\times 2$ matrices. We therefore only need to consider $M(\vec{\theta})_{1,1}$. We have
\begin{equation}
    \del_{\theta_{1}}\langle \vec{O}\rangle_{\rho_{\vec{\theta}}} =  {N\over 8}\left( {N\over 4}+1\right)\sin 2\theta_{1}\begin{pmatrix}
        1& -1&0&0
    \end{pmatrix}
\end{equation}
so only the upper $2\times 2$ block of $\text{Cov}_{\rho_{\vec{\theta}}}(\vec{O})$ is relevant. From a table of at-most-fourth moments of spin operators in Dicke states, one obtains
\begin{align}
    \text{Cov}_{\rho_{\vec{\theta}}}(\vec{O})_{1,1}&= \sin^{4}\theta_{1} g(N) + t(N)\cos^{2}\theta_{1}\sin^{2}\theta_{1}  - t(N)^{2}\sin^{4}\theta_{1} \nonumber \\
    \text{Cov}_{\rho_{\vec{\theta}}}(\vec{O})_{2,2}&= \left(1+\cos^{4}\theta_{1}\right)g(N) +t(N)\cos^{2}\theta_{1}\sin^{2}\theta_{1}  - 2h(N)\cos^{2}\theta_{1}  - t(N)^{2}\sin^{4}\theta_{1}\nonumber \\
    \text{Cov}_{\rho_{\vec{\theta}}}(\vec{O})_{1,2}&=\cos^{2}\theta_{1}\sin^{2}\theta_{1}\left( g(N)- t(N) \right) - h(N)\sin^{2}\theta_{1} +t(N)^{2}\sin^{4}\theta_{1}
\end{align}
where
\begin{align}
t(N)&:= {1\over 2}\left( {N\over 4} \right)\left( {N\over 4}+1\right) \nonumber \\
    g(N)&:= {1\over 8}\left( 3\left( {N\over 4}\right)^{4}+6\left( {N\over 4}\right)^{3}+\left( {N\over 4}\right)^{2}-2\left( {N\over 4}\right)\right) \nonumber \\
    h(N)&:= {1\over 8}\left( \left( {N\over 4}\right)^{4}+2\left( {N\over 4}\right)^{3}+3\left( {N\over 4}\right)^{2}+2\left( {N\over 4}\right)\right)
\end{align}
are scaled polynomials of $N/4$.
To evaluate $M(\vec{\theta})_{1,1}$, it is sufficient to note that the determinant of the upper left $2\times 2$ block of $\text{Cov}_{\rho_{\vec{\theta}}}(\vec{O})$ is
\begin{equation}
    {\sin^{2}(2\theta_{1}) \over 16}\left( {N\over 4}\right)^{2}\left({N\over 4}+1\right)^{2}\left({N\over 4}+2\right)\left({N\over 4}-1\right)
\end{equation}
and that 
\begin{equation}
    \text{Cov}_{\rho_{\vec{\theta}}}(\vec{O})_{1,1}+\text{Cov}_{\rho_{\vec{\theta}}}(\vec{O})_{2,2}+2\text{Cov}_{\rho_{\vec{\theta}}}(\vec{O})_{1,2}={1\over 2}\left( {N\over 4}\right)\left( {N\over 4}+2\right)\left({N\over 4}+1\right)\left({N\over 4}-1\right).
\end{equation}
One gets that
\begin{equation}
    M(\vec{\theta})_{1,1}= \left( {N\over 8}\left( {N\over 4}+1\right) \right)^{2}\sin^{2} 2\theta_{1} { \text{(C9)}\over \text{(C8)} } = {N\over 8}(N+4)
\end{equation}
independent of $\theta_{1}$.

\end{document}